\documentclass[journal]{IEEEtran}
\usepackage{amsmath,amsfonts, amsthm}
\usepackage{array}
\usepackage[caption=false,font=normalsize,labelfont=sf,textfont=sf]{subfig}
\usepackage{textcomp}
\usepackage{stfloats}
\usepackage{url}
\usepackage{verbatim}
\usepackage{graphicx}
\usepackage{placeins} 
\usepackage{microtype}
\usepackage[caption=false,font=normalsize,labelfont=sf,textfont=sf]{subfig}

\def\BibTeX{{\rm B\kern-.05em{\sc i\kern-.025em b}\kern-.08em
    T\kern-.1667em\lower.7ex\hbox{E}\kern-.125emX}}
\usepackage{balance}
\usepackage[ruled]{algorithm2e}
\usepackage{amssymb}

\def\gap{0.9ex}
\abovedisplayskip\gap
\belowdisplayskip\gap
\abovedisplayshortskip\gap
\belowdisplayshortskip\gap

\newtheorem{lemma}{Lemma}

\begin{document}
\onecolumn 
\thispagestyle{empty}
\noindent {\Huge IEEE Copyright Notice}

\bigskip

\noindent {\large \copyright~2026 IEEE. Personal use of this material is permitted. Permission from IEEE must be obtained for all other uses, in any current or future media, including reprinting/republishing this material for advertising or promotional purposes, creating new collective works, for resale or redistribution to servers or lists, or reuse of any copyrighted component of this work in other works.}
\newpage
\twocolumn 
\setcounter{page}{1}

\title{Adaptive Power Allocation and User Scheduling for LEO Satellites using Channel Predictions}
\author{
Lachlan~Drake, Lawrence~Ong, and Duy~T.~Ngo
\thanks{This article has been accepted for publication in IEEE TVT. This is the authors' version. The final published version will be available in IEEE Xplore. DOI: 10.1109/TVT.2026.3685307}
\thanks{L. Drake, L. Ong, and D.~T. Ngo are with the School of Engineering, The University of Newcastle, Callaghan NSW 2308 Australia.}
\thanks{This work is supported in part by the NSW Connectivity Innovation Network, and the Commonwealth Scientific \& Industrial Research Organisation.}
}

\markboth{IEEE Transactions on Vehicular Technology,~Vol.~XX, No.~X, XXX~2026}%
{Drake \MakeLowercase{\textit{et al.}}: Adaptive Power Allocation and User Scheduling for LEO Satellites using Channel Predictions}

\maketitle

\begin{abstract}
	Low earth orbit (LEO) satellites are a key technology to enable connectivity for rural and remote users. Communication satellites in LEO can provide coverage to much larger areas than terrestrial or aerial systems, while offering improved data rates when compared with geostationary systems. However, a major challenge with LEO satellite communications is the high mobility of the satellite, which results in a rapidly changing communication channel. Due to this, it is challenging to fairly allocate communication resources to multiple users in the system. This work proposes an Adaptive Power Allocation and Scheduling Scheme (APASS) to ensure user fairness in the downlink of a LEO satellite system serving mobile ground users. First, a novel channel and transmission model is introduced to capture the variability in channel statistics due to the satellite’s trajectory. Then, a non-convex optimization problem is formulated to maximize the minimum rate across all ground users over a fixed set of time slots. To solve this problem, the proposed APASS dynamically allocates power and schedules transmissions based on predicted future channel gains. Numerical results show that APASS achieves strong performance even with substantial prediction errors, faring close to an upper bound that assumes perfect future channel knowledge. Furthermore, it improves the minimum user rate by a factor of 2.98 compared to equal-power allocation and maintains user fairness with a Jain's fairness index of well above 0.99.
	
\end{abstract}


\section{Introduction}
\IEEEPARstart{W}{ith} the expected increased demand and expanded use cases for the sixth generation (6G) of wireless communication technology, the use of space-based communication systems is expected to increase. In particular, one of the key visions of 6G is global coverage to connect remote areas, where it is impractical to deploy terrestrial communication infrastructure. Satellite systems are key to enabling this vision, as they offer coverage over larger areas than terrestrial or aerial systems. Further, recent advances in launch technology have greatly decreased the cost of deploying satellites into low earth orbit (LEO). LEO satellites offer lower latency and can achieve higher data rates when compared with geostationary services. 

A major challenge with LEO satellite communications is the high mobility of the satellite, which results in a rapidly changing communication channel. This poses difficulties in fairly allocating communication resources to multiple users in the system. While highly varying channels are common in many scenarios, the known trajectory of satellites presents an opportunity to predict the channel thereby achieving better rate performance and fairness among the users.
Most prior works in resource allocation for satellite communications consider the channel as a stationary stochastic process for some period of time \cite{Letzepis2008}, \cite{You2020}, \cite{Gao2021}, but not one where the change in the relative positions of the transceivers is taken into account.

While some previous studies examine adaptive schemes, they generally either consider only single-step prediction or a statistical model without specific predictions. Hu et al. \cite{Hu2020} consider power control for a LEO constellation, but use a long-term statistical expectation for channel gains. Lu et al. \cite{Lu2025} consider adaptive power control in a satellite system, by using statistical models of the satellite elevation angles. Single-step prediction methods, such as Kalman filters, are often proposed for channel prediction in terrestrial and orbital systems alike \cite{Shoarinejad2003}, \cite{Vu2023}. However, most communications satellites operate in well-defined, predictable orbital trajectories. This deterministic orbital knowledge presents an opportunity to enhance power control schemes by enabling precise temporal predictions of channel gains.

In this paper, we present APASS, a novel framework that leverages channel gain predictions to dynamically allocate transmit power and schedule user transmissions. Unlike previous adaptive schemes, APASS provides a framework to schedule users and allocate power based on specific channel gain predictions multiple time steps ahead. The unique problem formulation uses channel predictions to optimize over the entire time horizon, but remains flexible by updating its allocations as accurate channel estimations are obtained. This allows the transmitter to better plan for and adapt to changes in the channel conditions, compared to schemes that only consider the current time slot or use only a single prediction. This applies especially to the highly-dynamic, yet predictable, LEO satellite channel.

This paper develops a space-ground channel and transmission model that incorporates the dynamic variations in channel parameters due to satellite motion, while also capturing the power spectrum characteristics associated with Doppler shift. Then, we formulate an optimization problem to maximize user fairness across several time slots in the downlink. This problem is intractable due to its non-convex structure and the stochastic nature of the channel. However, by using predictions of the future channel alongside successive convex approximation, we propose an adaptive power allocation and scheduling scheme to achieve a practical solution. Through numerical results, we show that the minimum user rate achieved by the scheme far exceeds schemes that do not use channel predictions. 

\section{LEO Satellite Downlink Channel Model}\label{sec:SysModel}
Consider a LEO satellite capable of transmitting highly directional spotbeams. Within each beam footprint, there are $K$ single-antenna ground user equipments (UEs). This work considers the power allocation for UEs within a single beam. The considered channel is adapted from the recommendations for non-terrestrial networks developed by the 3rd Generation Partnership Project (3GPP)~\cite{3GPP2020}. The channel is modeled with two states, representing line-of-sight (LoS) and non-LoS (NLoS) conditions. Both the large-scale parameters and the fading parameters vary with the satellite's elevation angle.

\subsection{Large-scale Parameters}

Let $R_\text{E}$ represent the radius of the Earth and $d_0$ the satellite altitude. The distance $d$ between the satellite and the UE is a function of the elevation angle $\alpha$ as 
\begin{align}
	d(\alpha) &= \sqrt{R_\text{E}^2\sin^2(\alpha)+d_0^2+2d_0R_\text{E}}-R_\text{E}\sin(\alpha).
\end{align}

Due to the movement of the satellite over great distances in a short period of time, the large-scale parameters of the channel are subject to significant change over the course of one visibility period. The total path loss in dB is given by
\begin{align}
	\Gamma_{\text{dB}}(\alpha) &=  \mathsf{FSPL}(\alpha) + \mathsf{SF}(\alpha) +\mathsf{CL}(\alpha) + \mathsf{PL}_\text{g}(\alpha),
 \end{align}
where $\mathsf{FSPL}(\alpha)$ is the free-space path loss, $\mathsf{SF}(\alpha)$ is shadow fading, $\mathsf{CL}(\alpha)$ is clutter loss, and $\mathsf{PL}_\text{g}(\alpha)$ is attenuation due to atmospheric gases. All values are in dB.
Specifically, $\mathsf{FSPL}(\alpha)=20\log_{10}(\frac{4\pi d(\alpha) f_\text{c}}{c})$ with the carrier frequency $f_\text{c}$, and the speed of light $c$. $\mathsf{SF}(\alpha)$ is a log-normally distributed random variable with zero mean and variance $\sigma_\text{SF}^2$ determined by the conditions specified by the 3GPP report~\cite{3GPP2020}. $\mathsf{SF}(\alpha_1)$ and $\mathsf{SF}(\alpha_2)$ are independent for distinct angles $\alpha_1$ and $\alpha_2$; they are also independent for distinct receivers. The clutter loss $\mathsf{CL}(\alpha)$ is caused by attenuation due to ground objects, with values for clutter loss at various elevation angles and frequencies given by the 3GPP report \cite{3GPP2020}. When there is a direct LoS path, $\mathsf{CL}(\alpha) = 0 $ dB. The atmospheric attenuation $\mathsf{PL}_\text{g}(\alpha)$ follows the ITU standards \cite[Annex 2]{ITU676}. 
 \subsection{Fading Parameters}
 At high elevation angles, there is often a LoS component in the received signal. However, it is also possible that the LoS path is blocked or shadowed. Our channel model uses a semi-Markov process with two states, representing LoS and NLoS conditions. The small-scale fading distribution is Rician in the LoS state \cite{Letzepis2008, You2020} and Rayleigh in the NLoS state.
 
 The fading (or small-scale) parameters $a(t)\in\mathbb{C}$ are generated by using a sum-of-sinusoids model \cite{Jakes1994}. This method ensures that the generated channel has a power spectrum and auto-correlation function that reflects the satellite's dynamics. $L$ sinusoids are summed, each with a random phase term $\theta_{l,i}$ which is independently
and identically distributed (i.i.d.) with a uniform distribution in the range $[0,~2\pi)$. The Doppler frequencies $\nu_{l,i}$ for each sinusoid are calculated based on the maximum Doppler shift $\nu_{\text{max}}$. Specifically, $a(t)= a_{1}(t) + ja_{2}(t)$ where
\begin{align}
 	a_{i}(t)&=\sqrt{\frac{2}{L}} \sum_{l=1}^L\cos(2\pi \nu_{l,i}t+\theta_{l,i}), ~i=1,2,\\
 	\nu_{l,i}&=\nu_{\text{max}}\cos\left(\frac{\pi}{2L}(l-0.5)+(-1)^{i-1}\frac{\pi}{12L}\right).
 \end{align}
  The maximum Doppler shift $\nu_{\text{max}}$ is based only on the orbital dynamics of the satellite. Any Doppler effect associated with UE movement is assumed negligible in comparison due to the much higher relative speed of the satellite. $L$ is chosen to control the fidelity of the simulation  \cite{Patzold2006}. Note that the channels for different UEs are generated independently.

 \subsection{Channel Gain}
 The fading parameters $a(t)$ are then scaled by the large-scale parameter $\Gamma(\alpha)$ to obtain the channel gain $h(t)\in\mathbb{C}$. For the Rayleigh (or NLoS) state, $h(t)=\frac{1}{\sqrt{\Gamma(\alpha)}}a(t)$.
 For the Rician (or LoS) state with $K$-factor $\kappa$ and random uniformly-distributed phase $\phi_r$, the channel gain is generated as
 \begin{align}
 	h(t) &= \frac{1}{\sqrt{\Gamma(\alpha)}}\left[\frac{a(t)}{\sqrt{\kappa+1}}+\sqrt{\frac{\kappa}{\kappa+1}}e^{j(2\pi\nu_{\text{max}}t+\phi_{r})}\right].
 \end{align}

 Note: although many of the channel parameters are determined by the satellite elevation angle $\alpha$, the channel is inherently stochastic. Its randomness is due to the phase terms $\theta_{l,i}$ in small-scale fading, and its stochasticity arises from the shadow-fading term $\mathsf{SF}(\alpha)$ in large-scale fading.

\subsection{LoS Probability}
The LoS conditions (whether a UE is in the LoS or NLoS state at any time) are independent among the UEs and over time. Their probability distributions are based on the UE environment (e.g., urban or rural) and the elevation angle, as stated in the 3GPP report \cite{3GPP2020}.
In the LoS state, the small-scale coefficients are generated with a Rician distribution, there is no clutter loss, and the shadow fading distribution has significantly lower variance.

\section{LEO Satellite Transmission Model}
Consider a set of UEs indexed by $k\in\mathcal{K} \triangleq \{1,\dotsc,K\}$. In discrete time, let $N$ denote the number of coherence intervals under consideration and $T$ the number of transmitted symbols per coherence time $T_\text{coh}$. For a UE $k$ at symbol time $i\in\{1,\dotsc,NT\}$, let $h_{k,i}$ denote the discrete channel gain, $x_{i}\in\mathbb{C}$ the transmitted symbol, and $z_{k,i}$ the additive noise.  The received signal is then $y_{k,i} = h_{k,i}x_{i}+z_{k,i}$.

Recall that the channel gains among the UEs, $\{h_{k,i}: k \in \mathcal{K}\}$, are mutually independent, but for each UE~$k$, its channel gains over time $\{h_{k,i}: i \in \{1, \dotsc, NT\}\}$ are correlated.  
The additive noise is i.i.d with distribution $z_{k,i}\sim\mathcal{CN}(0,\sigma^2)$. The noise variance is $\sigma^2 = k_{\text{B}}T_{\text{n}}W$, where $k_{\text{B}}$ is the Boltzmann constant, $T_{\text{n}}$ is the noise temperature, and $W$ is the system bandwidth.

It is assumed that the channel gain is constant for a coherence time of $T$ symbols. So, the $NT$ transmissions can be viewed as $N$ time slots of length $T$ each, such that for each time slot $n \in \{1,\dotsc, N\}$, we have $h_{k,i}=\bar{h}_{k,n} ~\forall i\in \{(n-1)T+1,\dotsc, nT\}$.

Assume that the transmitter knows the channel gains $\bar{h}_{k,n}$ for each time slot~$n$ only at the start of the time slot. 
We impose a per time slot power constraint of $\frac{1}{T}\sum_{i=(n-1)T}^{nT} |x_i|^2 \leq P_\text{total}$ for each~$n$.

The satellite transmits the superposition of signals for multiple users across a single channel, with the transmission for each UE $k$ in time slot $n$ having power $\bar p_{k,n}$.  For each UE, the transmissions to all other UEs are treated as interference, i.e., no interference cancellation methods are employed. Therefore, the signal-to-interference-plus-noise ratio (SINR) for the downlink is a function of the channel gain for a UE~$k$ at slot $n$, i.e., $\bar{h}_{k,n}$ and all power allocations for slot $n$, i.e., $\boldsymbol{\bar{p}}_{n}= \begin{bmatrix} \bar{p}_{1,n} & \bar{p}_{2,n} & \dotsc & \bar{p}_{K,n} \end{bmatrix}^\text{T}$. Let the channel power gain be denoted by $g_{k,n} =\lvert\bar{h}_{k,n}\rvert^2$.
The SINR for a UE~$k$ in a time slot $n$ is then
\begin{equation}
    \text{SINR}_{k,n} \left( \boldsymbol{\bar{p}}_{n},\bar{g}_{k,n} \right) = \frac{g_{k,n} \bar{p}_{k,n}}{g_{k,n} \underset{j\in \mathcal{K}\setminus\{k\}}{\sum} \bar{p}_{j,n}+\sigma_{k,n}^2}. \label{eq:sinr}
\end{equation}

For a sufficiently large number of samples $T$, the achievable rate for each UE~$k$ and each time slot $n$ in bits/s/Hz is arbitrarily close to the following:
\begin{equation}
	R_{k,n} \left( \boldsymbol{\bar{p}}_{n}, g_{k,n} \right)= \log_2\left( 1+\text{SINR}_{k,n} \left( \boldsymbol{\bar{p}}_{n},g_{k,n} \right) \right). \label{eq:slot-rate}
\end{equation}

Denote the $K\times N$ power allocation matrix and the UE $k$'s $1\times N$ channel gain vector as
\begin{align*} 
	\boldsymbol{\bar{P}} &= 
	\begin{bmatrix}
		\bar{p}_{1,1} & \dotsc & \bar{p}_{1,N}\\
		\vdots & & \vdots\\
		\bar{p}_{K,1} & \dotsc & \bar{p}_{K,N}
	\end{bmatrix}, ~~
	\boldsymbol{g}_{k} &= 
	\begin{bmatrix} 
		g_{k,1} & \dotsc & g_{k,N}
	\end{bmatrix}.
\end{align*}
The total rate for UE~$k$ after a transmission block of $N$ slots is then

\begin{equation}
    R_{k} \left( \boldsymbol{\bar{P}}, \boldsymbol{g}_k \right)=\frac{1}{N}\sum_{n=1}^{N} R_{k,n} \left( \boldsymbol{\bar{p}}_n, \boldsymbol{g}_{k} \right). \label{eq:total-rate-def}
\end{equation}

\section{Problem Formulation}

This paper aims to find an optimal power allocation for fairness among the UEs, i.e. finding $\boldsymbol{\bar{P}}$ to maximize the minimum achievable rate across $N$ slots for all $K$ UEs. Recall that $\boldsymbol{g}_n$ is known to the transmitter and the receiver only at the beginning of each slot $n$. As such, $\boldsymbol{p}_n$ can only be chosen based on the known channel gain up to slot $n$, and the power allocation for the previous slots. Denote this information as
\begin{equation}
	\boldsymbol{A}_{n} = 
	\begin{bmatrix}
		g_{1,1} & \dotsc & g_{1,n} & \bar{p}_{1,1} & \dotsc  & \bar{p}_{1,(n-1)}\\
		\vdots & & & & & \vdots\\
		g_{K,1} & \dotsc & g_{K,n} & \bar{p}_{K,1} & \dotsc  & \bar{p}_{K,(n-1)}
	\end{bmatrix}.
	\end{equation}

 The power allocation for each UE in slot $n$ is thus given by the following function (which can be deterministic or non-deterministic): $f_{k,n}: \mathbb{C}^{K\times n}\times \mathbb{R}_0^{+^{K\times (n-1)}} \rightarrow \mathbb{R}_0^+$, i.e.,
\begin{equation}
	\bar{p}_{k,n} = f_{k,n}\left( \boldsymbol{A}_{n} \right).\\
 \label{pbar}
\end{equation}
For any chosen power-allocation scheme $\boldsymbol{f} := [f_{k,n}]_{k=1,\dots,K; \, n=1,\dots,N}$, the achievable rates across $N$ slots are random variables that are functions of the channel gains. Hence, the aim is to optimize the functions to maximize the expected value of the minimum UE rate, where the expectation is taken with respect to the channel gains $\boldsymbol{g} := [g_{k,n}]_{k=1,\dots,K; \, n=1,\dots,N}$. This observation results in the following max-min problem:
\begin{subequations}\label{eq:ProbForm}
	\begin{align}  
        R^* = \max_{\boldsymbol{f}}~ 
        & \mathbb{E}_{\boldsymbol{g}}\left[\underset{i\in \{1,\dotsc, K\}}{\min} ~ R_{i} \left( \boldsymbol{\bar{P}}, \boldsymbol{g}_{i} \right) \right] \label{subeq:ProbForm_Obj} \\
		\textrm{s.t.}~~&\sum_{k=1}^{K} \bar{p}_{k,n} \leq P_{\text{total}} \quad \forall n \\
		&\bar{p}_{k,n} \geq 0 \quad \forall k, \forall n.
	\end{align}
\end{subequations}

\section{APASS: Adaptive Power Allocation and Scheduling Scheme}
Problem \eqref{eq:ProbForm} is a stochastic, non-convex problem which is very challenging to solve.
In the following, we first derive an upper bound to the optimal value $R^*$ of \eqref{eq:ProbForm} by using a genie-aided argument and approximate it using a successive convex approximation. For achievability, we then propose an adaptive power allocation and scheduling scheme (APASS). By setting UE transmit powers to zero in certain time slots, APASS performs UE scheduling. While constituting a lower bound to the optimal value $R^*$ of \eqref{eq:ProbForm}, APASS offers a practical solution to an otherwise intractable problem.

\subsection{Genie-Aided Upper Bound and Successive Convex Approximation} 
Let us consider the genie-aided scenario, in which all \emph{past} and \emph{future} channel gains are known to the transmitter at every slot $n$. For a specific channel realization $\boldsymbol{g}$, we define the following, which is the minimum UE rate achieved by the optimal power allocation when $\boldsymbol{g}$ is known at the beginning:
 \begin{subequations}\label{eq:GenieProb}
	\begin{align}
		R^\text{u} (\boldsymbol{g}) = \underset{\boldsymbol{\bar{P}}}\max & \underset{k\in \{1,\dotsc, K\}}{\min} \left\{R_{k}\left(\boldsymbol{\bar{P}}, \boldsymbol{g}_{k}\right)\right\},\label{subeq:Genie_Obj}\\
		\textrm{s.t.}~~ & \sum_{k=1}^{K} \bar{p}_{k,n} \leq P_{\text{total}} \quad \forall n\label{eq:power-constraint1}\\
		& \bar{p}_{k,n} \geq 0 \quad \forall k, \forall n.
	\end{align}
\end{subequations}

Denote the minimum rate attained with power-allocation scheme~$\boldsymbol{f}$ over the channel gain realization $\boldsymbol{g}$ by $R(\boldsymbol{f},\boldsymbol{g}) := \underset{i\in \{1,\dotsc, K\}}{\min} ~ R_{i} \left( \boldsymbol{\bar{P}}, \boldsymbol{g}_{i} \right)$, where $\boldsymbol{\bar{P}}$ is a function of $\boldsymbol{f}$ according to \eqref{pbar}.
Also, let $\boldsymbol{f}^*$ be a power-allocation scheme that attains $R^*$. Clearly, $R^\text{u}(\boldsymbol{g}) \geq R(\boldsymbol{f},\boldsymbol{g})$ for any $\boldsymbol{f}$ and any $\boldsymbol{g}.$ Consequently, $\mathbb{E}_{\boldsymbol{g}}[R^\text{u}(\boldsymbol{g})] \geq \mathbb{E}_{\boldsymbol{g}}[R(\boldsymbol{f}^*,\boldsymbol{g})] = R^*$.

Problem \eqref{eq:GenieProb} is still non-convex due to \eqref{subeq:Genie_Obj}. In what follows, we will solve it via successive convex approximation where each iteration involves solving a (convex) geometric program (GP) only. 
First, we transform \eqref{eq:GenieProb} into the following equivalent form:
\begin{subequations}\label{eq:EquivProb2}
	\begin{align}
		\underset{\boldsymbol{\bar{P}}}{\min} \quad & b\\
		\textrm{s.t.}~~ &\prod_{n=1}^N \frac{ g_{k,n} \underset{j\in \mathcal{K}\setminus\{k\}}{\sum} \bar{p}_{j,n}+\sigma_{k,n}^2 }{ g_{k,n} \underset{l\in \mathcal{K}}{\sum} \bar{p}_{l,n}+\sigma_{k,n}^2}\leq b \quad \forall k \label{subeq:Equiv_con1}\\
		&\sum_{k=1}^{K} \bar{p}_{k,n} \leq P_{\text{total}} \quad \forall n\\
		&\bar{p}_{k,n} \geq 0 \quad \forall k, \forall n.
	\end{align}
\end{subequations}

Let us now define $\gamma_{k}(\boldsymbol{\bar{P}})\triangleq \prod_{n=1}^N g_{k,n} \underset{j\in \mathcal{K}\setminus\{k\}}{\sum} \bar{p}_{j,n}+\sigma_{k,n}^2 $. Also define $\lambda_{k}(\boldsymbol{\bar{P}}) \triangleq \prod_{n=1}^N \zeta_{k,n}(\boldsymbol{\bar{P}})$, where $\zeta_{k,n}(\boldsymbol{\bar{P}})=g_{k,n} \underset{l\in \mathcal{K}}{\sum} \bar{p}_{l,n}+\sigma_{k,n}^2$. Then, the LHS of \eqref{subeq:Equiv_con1} is simply $\mathsf{ISNR}_{{k}}(\boldsymbol{\bar{P}}) = \frac{\gamma_k(\boldsymbol{\bar{P}})}{{\lambda}_k(\boldsymbol{\bar{P}})}$. To enable geometric programming, we need ${\mathsf{ISNR}_{k}}(\boldsymbol{\bar{P}})$ to be a posynomial\footnote{A function $f(\boldsymbol{x}) = c x_1^{a_1} x_2^{a_2} \dotsc x_n^{a_n}$, where $c>0$, $a_i\in\mathbb{R}$, $\boldsymbol{x}\in\mathbb{R}^n_{++}$ is called a monomial. A posynomial is defined as a sum of monomials \cite{Boyd2023}.}. Since $\gamma_{k}(\boldsymbol{\bar{P}})$ is already a posynomial, we only need to approximate posynomial $\lambda_{k}(\boldsymbol{\bar{P}})$ by a monomial by using an arithmetic-geometric mean approximation. This approximation is iteratively refined through the successive convex approximation framework outlined in Algorithm~\ref{alg:GeomProg}. 
{In the $i$th iteration, we solve a (convex) GP over $\boldsymbol{\bar{P}}$. The solution, denoted by $\boldsymbol{\bar{P}}^{[i]}$, is used to compute coefficients for the GP over $\boldsymbol{\bar{P}}$ in the next $(i+1)$th iteration.}

Specifically, at the $i$th iteration, we approximate the posynomial $\zeta_{k,n}(\boldsymbol{\bar{P}})$ by the following monomial:
\begin{equation}\label{eq:ApproxZeta}
	\tilde{\zeta}_{k,n}^{[i]}(\boldsymbol{\bar{P}}) = \prod_{m=1}^{K+1} \left(\frac{u_{k,n}^{(m)}(\boldsymbol{\bar{P}})}{w_{k,n}^{(m),[i]}}\right)^{w_{k,n}^{(m),[i]}},
\end{equation}
where $u_{k,n}^{(m)}(\boldsymbol{\bar{P}})$ is defined as the $m$th monomial component of the posynomial $\zeta_{k,n}(\boldsymbol{\bar{P}})$:
\begin{equation*}
	u_{k,n}^{(m)}(\boldsymbol{\bar{P}})=\begin{cases}
		g_{k,n}\bar{p}_{m,n}, & \text{if } m = 1, \dotsc, K\\
		\sigma_{k,n}^2, & \text{if } m = K + 1,
	\end{cases}
\end{equation*}
and the weights $w_{k,n}^{(m),[i]}$ are calculated by using $\boldsymbol{\bar{P}}^{[i-1]}$  of the previous iteration as
\begin{equation}\label{eq:weights}
	w_{k,n}^{(m),[i]} = \frac{u_{k,n}^{(m)}(\boldsymbol{\bar{P}}^{[i-1]})}{\zeta_{k,n}(\boldsymbol{\bar{P}}^{[i-1]})}.
\end{equation}
Then, $\lambda_{k}(\boldsymbol{\bar{P}})$ is approximated by the monomial $\tilde{\lambda}^{[i]}_k(\boldsymbol{\bar{P}}) = \prod_{n=1}^{N} \tilde{\zeta}^{[i]}_{k,n}(\boldsymbol{\bar{P}})$, and the resulting approximation $\widehat{\mathsf{ISNR}}^{[i]}_k(\boldsymbol{\bar{P}}) = \frac{\gamma_k(\boldsymbol{\bar{P}})}{\tilde{\lambda}^{[i]}_k(\boldsymbol{\bar{P}})}$ is finally a posynomial. At the $i$th iteration, instead of solving problem \eqref{eq:EquivProb2}, we solve its approximated version:
\begin{subequations}\label{eq:EquivProb3}
	\begin{align}
		\underset{\boldsymbol{\bar{P}}}{\min} \quad & b\\
		\textrm{s.t.}~~ &\widehat{\mathsf{ISNR}}^{[i]}_k(\boldsymbol{\bar{P}}) \leq b \quad \forall k \label{subeq:Equiv_con3}\\
		&\sum_{k=1}^{K} \bar{p}_{k,n} \leq P_{\text{total}} \quad \forall n\\
		&\bar{p}_{k,n} \geq 0 \quad \forall k, \forall n.
	\end{align}
\end{subequations} 
Problem \eqref{eq:EquivProb3} is a (convex) GP, which can be solved very efficiently by any convex optimization solver, e.g., using the interior-point method \cite{Boyd2023}. Algorithm~\ref{alg:GeomProg} iterates until there is no improvement in the solution. It can be shown that the solution obtained at the convergence of Algorithm~\ref{alg:GeomProg} satisfies the Karush-Kuhn-Tucker conditions for problem \eqref{eq:EquivProb2} (and hence problem \eqref{eq:GenieProb}), indicating that the necessary conditions for optimality are fulfilled  \cite{Chiang2007}.

\begin{algorithm}[t]\label{alg:GeomProg}
	Initialize $\boldsymbol{\bar{P}}^{[0]}$ and $i=1$\;
	\Repeat{convergence}{
		Calculate weights $w_{k,n}^{(m),[i]}$ $\forall k, \forall n, \forall m$, by \eqref{eq:weights}\;
        Evaluate $\tilde{\zeta}^{[i]}_{k,n}(\boldsymbol{\bar{P}}) ~~\forall k, \forall n$, by \eqref{eq:ApproxZeta}\;
		Evaluate $\tilde{\lambda}^{[i]}_k(\boldsymbol{\bar{P}}) = \prod_{n=1}^{N} \tilde{\zeta}^{[i]}_{k,n}(\boldsymbol{\bar{P}}) ~~\forall k$\;
		Evaluate $\widehat{\mathsf{ISNR}^{[i]}_{k}}(\boldsymbol{\bar{P}}) = \frac{\gamma_k(\boldsymbol{\bar{P}})}{\tilde{\lambda}^{[i]}_k(\boldsymbol{\bar{P}})} ~~\forall k$\;
		Compute $\boldsymbol{\bar{P}}^{[i]}$ by solving geometric program \eqref{eq:EquivProb3}\;
		Set $i=i+1$\;
	}
	\caption{Successive convex approximation for problem \eqref{eq:GenieProb}.}
\end{algorithm}

\subsection{Proposed Achievable Scheme}
In \eqref{eq:ProbForm}, the channel gains are not known ahead of time. For such practical scenarios, we propose an achievable transmission scheme (called APASS) in which predictions of the \emph{future} channel gains are used to solve for power allocation at a current time slot. Since APASS does not have perfect knowledge of channel gains at all times, it offers a lower bound to the optimal value $R^*$ of the original problem \eqref{eq:ProbForm}.
Several existing methods such as deep learning \cite{Zhang2021}, auto-regression, support vector machines, and Bayesian inference \cite{Liao2015} can be used for channel prediction. However, the development of specific channel prediction methods is out of the scope of this paper, where we focus on the resource allocation aspect instead.

Specifically, the proposed APASS runs at the beginning of each time slot $n\in\{1,\dotsc,N\}$, at which point the transmitter has already committed to power allocations $\boldsymbol{\bar p}_{[1:n-1]} :=[\boldsymbol{\bar p}_1 \dotsm \boldsymbol{\bar p}_{n-1}]$ and transmitted messages in \emph{previous} time slots $\{1,\dotsc, n-1\}$. It now acquires the actual channel gains $g_{k,n}$ for the \emph{current} time slot~$n$ and predicts channel gains $\tilde g_{k,l}$ for \emph{future} time slots $l \in \{n+1, \dotsc, N\}$. With this information, APASS solves the power allocations $\boldsymbol{\bar p}_{[n:N]} :=[\boldsymbol{\bar p}_n \dotsm \boldsymbol{\bar p}_{N}]$ for the current time slot $n$ and future time slots $\{n+1,\dotsc, N\}$: 
\begin{subequations}\label{eq:Scheme1}
	\begin{align}
	R_n^\text{A} (\boldsymbol{\bar p}_{[1:n-1]},\boldsymbol{g}) &=	\underset{\boldsymbol{\bar p}_{[n:N]}}{\max}  \underset{k\in \{1,\dotsc, K\}}{\min}\Biggl\{\frac{1}{n}\sum_{m=1}^{n}R_{k,m}\left( \boldsymbol{\bar{p}}_m, g_{k,m} \right) \notag \\ 
        & \quad + \frac{1}{N-n}\sum_{l=n+1}^{N}R_{k,l}\left(\boldsymbol{\bar{p}}_{l}, \tilde{g}_{k,l} \right)\Biggr\},\label{subeq:Scheme1_Obj}\\
		&\textrm{s.t. }\sum_{k=1}^{K} \bar{p}_{k,n} \leq P_{\text{total}} \quad \forall n,\\
		&\quad~~
        \bar{p}_{k,n} \geq 0 \quad \forall k, \forall n.
	\end{align}
\end{subequations}
Once a solution $\boldsymbol{\bar p}_{[n:N]}$ is found for time slot $n$, the transmitter commits only to $\boldsymbol{\bar p}_n$ and uses this power allocation to transmit messages in the time slot~$n$. The proposed APASS is then repeated at the beginning of the next time slot~$n+1$. This scheme is illustrated in Fig.~\ref{fig:schematic}.

Since problem \eqref{eq:Scheme1} has the same structure as that of problem \eqref{eq:GenieProb}, Algorithm~\ref{alg:GeomProg} can be straightforwardly modified to solve the former. Specifically, the approximation $\widehat{\mathsf{ISNR}}^{[i]}_k(\boldsymbol{\bar{P}})$ is modified according to the form of \eqref{subeq:Scheme1_Obj}. In each iteration of the modified version of Algorithm~\ref{alg:GeomProg}, a modified GP is solved.
\begin{lemma}
		 If $\tilde{g}_{k,n} = {g}_{k,n},~\forall k, ~\forall n$, then APASS, which is \eqref{eq:Scheme1} run successively over $N$ time slots, attains  $R^\text{u} (\boldsymbol{g})$.
	\end{lemma}
	
	\begin{IEEEproof}
		When $n = 1$, $R_{1}^\text{A} (\emptyset,\boldsymbol{g})=R^\text{u} (\boldsymbol{g})$. In the $n$th iteration of APASS, $\forall n \in \{1,\dots, N\}$, denote $\boldsymbol{\bar p}^{(n)}_{[1:n-1]} :=[\boldsymbol{\bar p}^{(n)}_1 \dotsm \boldsymbol{\bar p}^{(n)}_{n-1}]$ as the already committed power allocation and $\boldsymbol{\bar p}^{(n)}_{[n:N]}:=[\boldsymbol{\bar p}^{(n)}_n \dotsm \boldsymbol{\bar p}^{(n)}_{N}]$ as the obtained solution for \eqref{eq:Scheme1}. We claim that $R_{n+1}^\text{A} ([\boldsymbol{\bar p}^{(n)}_{[1:n]},\boldsymbol{g}) \geq R_{n}^\text{A} ([\boldsymbol{\bar p}^{(n)}_{[1:n-1]},\boldsymbol{g})$; otherwise, choosing $[\boldsymbol{\bar p}^{(n)}_{n+1} \dotsm \boldsymbol{\bar p}^{(n)}_{N}]$ as a solution for the $(n+1)$th iteration can strictly improve $R_{n+1}^\text{A} ([\boldsymbol{\bar p}^{(n)}_{[1:n]},\boldsymbol{g})$, which results in a contradiction. Since the value of \eqref{subeq:Scheme1_Obj} cannot decrease in each iteration, we obtain the required result by induction.
\end{IEEEproof}

\begin{figure}[!t]
	\centering
	\includegraphics[width=0.34\textwidth]{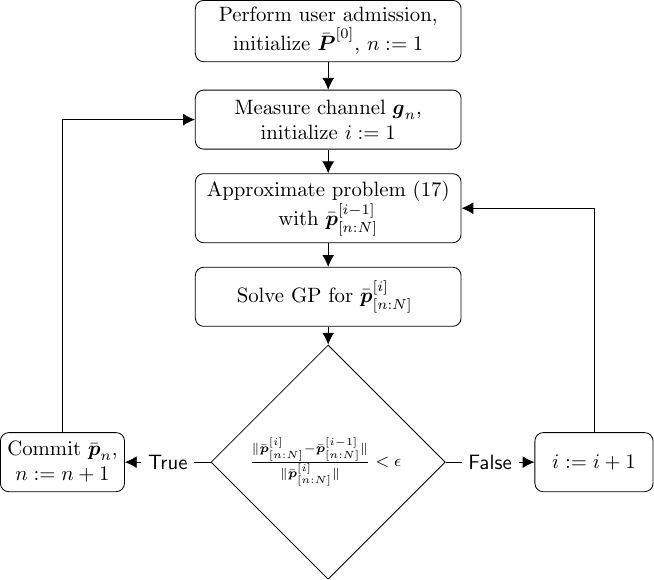}
	\caption{Schematic diagram of the APASS successive convex approximation algorithm. $\lVert\cdot\rVert$ is the Frobenius norm. $\epsilon$ is a small error tolerance value.}
	\label{fig:schematic}
\end{figure}

\vspace{-3mm}
\subsection{Complexity Analysis}
The APASS algorithm has two main components: (i) the problem-building component which includes calculating weights $w_{k,n}^{(m), [i]}$ and evaluating the approximated posynomial $\mathsf{ISNR}_{{k}}(\boldsymbol{\bar{P}})$; and
(ii) the problem-solving component, i.e., the interior-point solver used to evaluate the GP. The solver dominates the time complexity. Using an interior-point method for exponential conic programs typically requires solution of a linear system. Solving this linear system is the most expensive operation and has a time complexity of $\mathcal{O}(\nu^3)$ where $\nu$ represents the dimensions of the coefficient matrix \cite{Nocedal2006}. In our case $\nu$ is proportional to $K \times N$, therefore each solver call has complexity $\mathcal{O}(K^3 N^3)$. For each time slot $n$, the solver is called $I_n$ times, until convergence. Assuming $I := \max_n I_n$, the worst-case complexity of APASS is then $\mathcal{O}(I K^3 N^4)$. However, the practical complexity is generally much lower, as modern solvers can exploit parallelization as well as the sparsity and structure of the matrix to simplify the problem \cite{Dahl2021}.
\vspace{-3mm}

\subsection{Channel predictions}\label{sec:pred}
The focus of this paper is designing an effective power allocation and scheduling scheme to take advantage of channel predictions. Importantly, the focus is not on proposing or implementing these predictions, as this has been covered well in the existing literature. Traditional approaches have used statistical models to predict channel gains \cite{Fontan2001}. However, recent advancements in deep learning have also demonstrated significant promise in modeling and predicting channel gains in dynamic satellite systems.

For instance, Zhang et al. \cite{Zhang2022} introduced a deep learning-based method for downlink channel prediction in LEO systems. Similarly, Ying et al. \cite{Ying2024} proposed a convolutional neural network (CNN) and long short-term memory (LSTM)-based framework for joint channel prediction in LEO satellite systems, showing robust performance under high Doppler shifts and propagation delays. Recently, Zhao et al. \cite{Zhao2025} have presented a prediction method for the satellite downlink channel which combines CNN and transformer architectures and outperforms both a LSTM model and a transformer-only model.

\section{Numerical Results}\label{sec:Results}
Simulations are run with the parameters given in Table~\ref{tab:Parameters}. The satellite position is based on publicly available orbital element data for a SpaceX Starlink satellite. UE coordinates are normally distributed about a geographic point under the satellite's path near Canberra, Australia. All simulated UEs are located within a 50 km radius from the point under the satellite's path. The spotbeam's effective isotropic radiated power (EIRP) is used as the total power constraint $P_\text{total}$, with a conversion from Watts to Joules/sample. Note that the satellite beam is highly directional, and the EIRP is the equivalent power needed for the same received power using an isotropic antenna, not the actual power output of the satellite.

\begin{table}[!t]
	\begin{center}
		\caption{Simulation parameters and constants.}
		\label{tab:Parameters}
		\begin{tabular}{| c | c | c | c |}
			\hline
			Parameter & Symbol & Value & Unit\\
			\hline
			No. of ground UEs & $K$ & 20 & \\
			Coherence time &  $T_\text{coh}$ & 5 &s\\
			No. of time slots considered  &  $N$ & 30& \\
			Carrier frequency & $f_\text{c}$ & 27.5 &GHz\\		
			Radius of Earth& $R_\text{E}$ & 6371 &km\\			
			Orbital altitude&$d_0$ & 540 &km\\ 			
            No. of sinusoids (fading)& $L$ & 10&\\			 
			Effective isotropic radiated power & $\text{EIRP}$ & 5 &MW \\			
			Noise temperature & $T_\text{n}$ & 290 &K \\			
			System bandwidth & $W$ & 5 &MHz\\

			\hline
		\end{tabular}
        \vspace*{-3ex}
	\end{center}

\end{table}

First, to evaluate the effectiveness of the proposed APASS with varying levels of channel prediction accuracy, we generate the predicted channel gains by adding a complex normal error term to the actual channel gains, i.e., $\tilde{h}_{k,l} = h_{k,l} + \delta_{k,l} $, where $\delta_{k,l} \sim\mathcal{CN}(0,\sigma_{\text{e}}^2)$. As explained in Section \ref{sec:pred}, we do not propose a specific prediction method, but instead use Gaussian errors to illustrate the effect of prediction error on our algorithm. Fig.~\ref{fig:Results1} plots the results of APASS (by solving \eqref{eq:Scheme1}) as a fraction of the upper bound result (by solving \eqref{eq:GenieProb}) against the error variance. The genie-aided upper bound represents the case of perfect channel predictions. Fig.~\ref{fig:Results1} also shows the fairness of the UE rates as measured by Jain's fairness index \cite{Jain1984}. Here, the expected trend of decreasing effectiveness with increasing prediction error is clearly shown. It is also seen that as error increases, the achieved minimum rate deviates further from the apparent trend. The proposed APASS performs well even in cases with a significant prediction error, i.e., error variance equal to a quarter of the mean channel gain. In particular, a minimum rate above 80\% of the upper bound is achieved, while the fairness index remains well above 0.99.

\begin{figure*}[!t]
\centering
\subfloat[]{\includegraphics[height=0.32\textwidth]{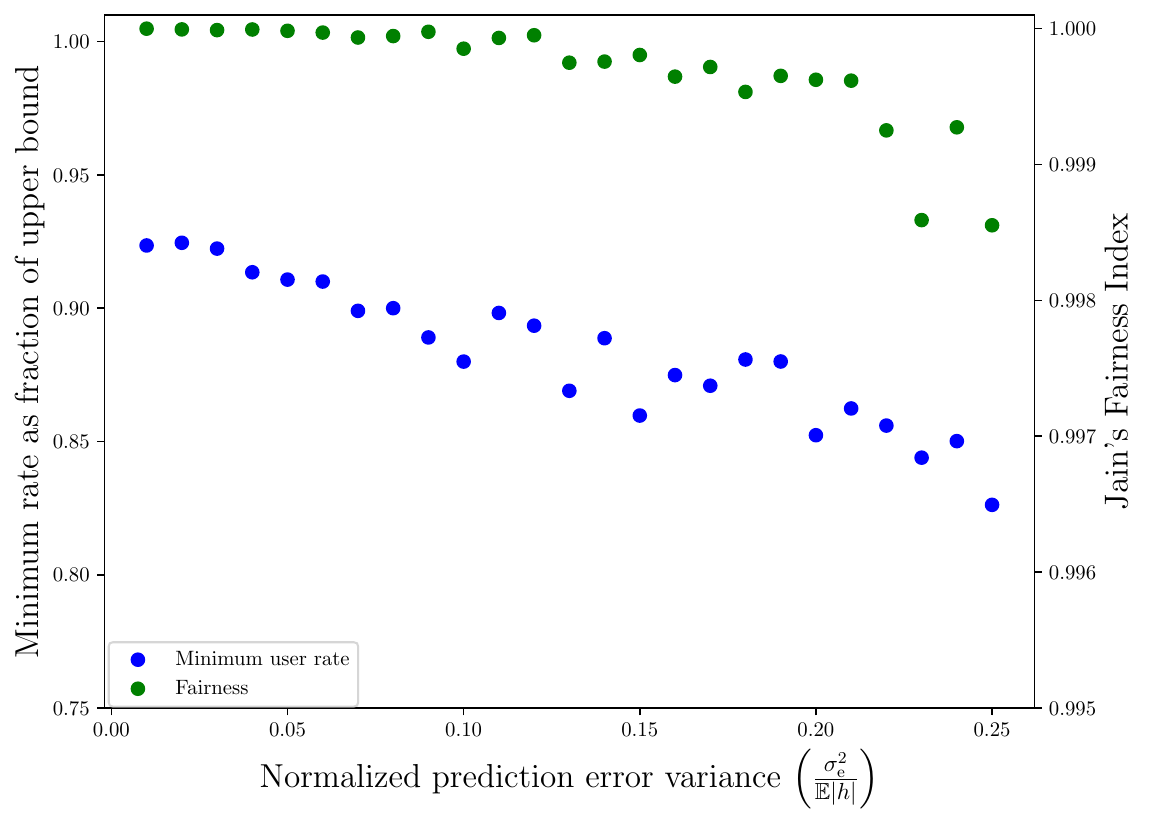}
\label{fig:Results1}}
\hfil
\subfloat[]{\includegraphics[height=0.32\textwidth]{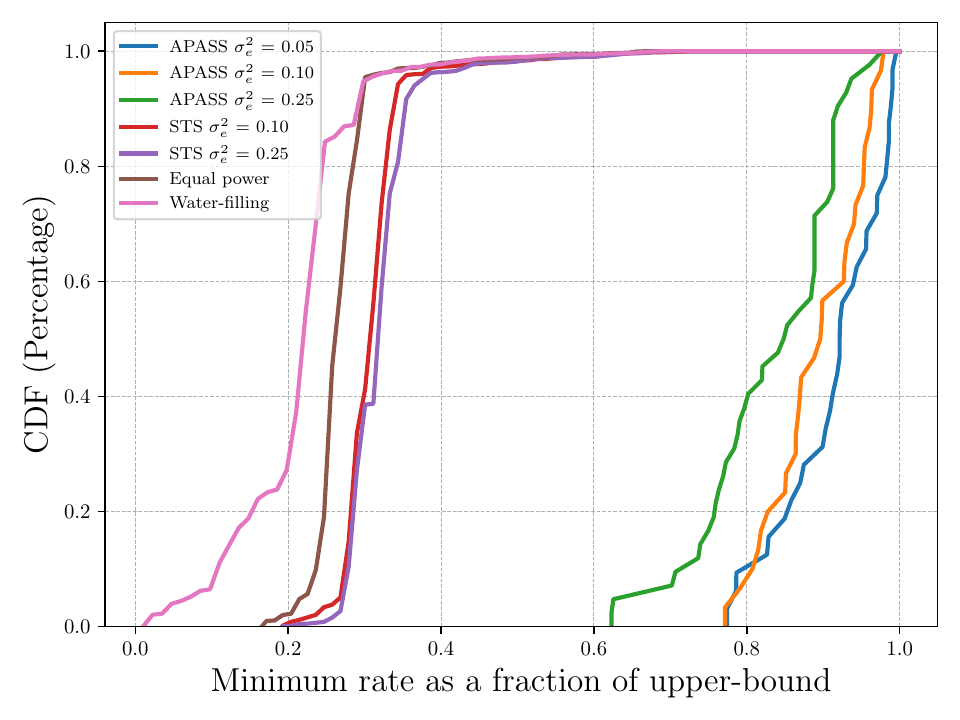}
\label{fig:Results3}}
\caption{(a) The minimum UE rate achieved by APASS as a percentage of the upper bound result and the scheme's fairness index plotted against the normalized variance of the channel prediction error. (b) Comparing cumulative distribution functions. The proposed APASS is simulated with a normalized prediction error variances of 0.05, 0.10 and 0.25. The STS scheme predicts one time slot ahead with error variance of 0.10 and 0.25. The equal-power scheme transmits to all UEs with equal power. The water-filling algorithm assigns power based on channel gains in each time slot.}
\label{fig:Results}

\end{figure*}

 Then, we compare APASS against three other achievable schemes that also operate with no future knowledge of the channel. The first scheme, called single-time-step (STS) scheme, has the same objective of max-min rates, but only predicts the future channel gain one time step ahead. 
 The second scheme, called the equal-power scheme, transmits equal power to all UEs for all time slots. 
 The third scheme uses the well-known water-filling algorithm \cite{Proakis2007} to allocate power between users in each time slot. 
 The time complexities of the STS, equal-power, and water-filling schemes are $\mathcal{O}(K^{3}N)$, $\mathcal{O}(K)$, and $\mathcal{O}(KN\log(K))$, respectively. To ensure fair comparison, all schemes are evaluated on the same simulated channels with $N$ time slots each. Fig.~\ref{fig:Results3} shows the cumulative distribution functions (CDF) of the minimum rate as a fraction of the upper-bound. The proposed APASS vastly outperforms the other schemes, even with considerable prediction error variance. Significantly, APASS with $\sigma^2_e = 0.25$ provides on average a minimum UE rate that is $2.30$, $2.98$, and  $3.86$ times higher than the STS, equal-power, and water-filling schemes, respectively.

\bibliography{Bibliography} 

\begin{thebibliography}{10}

\bibitem{Letzepis2008}
N.~Letzepis and A.~J. Grant, ``Capacity of the multiple spot beam satellite
  channel with {R}ician fading,'' {\em IEEE Trans. Inf. Theory}, vol.~54, Nov.
  2008.

\bibitem{You2020}
L.~You, K.-X. Li, J.~Wang, X.~Gao, X.-G. Xia, and B.~Ottersten, ``Massive
  {MIMO} transmission for {LEO} satellite communications,'' {\em IEEE J. Sel.
  Areas Commun.}, vol.~38, pp.~1851--1865, Aug. 2020.

\bibitem{Gao2021}
Z.~Gao, A.~Liu, C.~Han, and X.~Liang, ``Sum rate maximization of massive {MIMO}
  {NOMA} in {LEO} satellite communication system,'' {\em IEEE Wireless Commun.
  Lett.}, vol.~10, no.~8, pp.~1667--1671, 2021.

\bibitem{Hu2020}
J.~Hu, G.~Li, D.~Bian, L.~Gou, and C.~Wang, ``Optimal power control for
  cognitive {LEO} constellation with terrestrial networks,'' {\em IEEE
  Communications Letters}, vol.~24, no.~3, pp.~622--625, 2020.

\bibitem{Lu2025}
S.~Lu, G.~Xu, L.~Zhu, Z.~Song, and W.~Zhang, ``An adaptive power allocation for
  {NOMA}-based multi-layer satellite system,'' {\em IEEE Transactions on
  Wireless Communications}, vol.~24, no.~6, pp.~5269--5281, 2025.

\bibitem{Shoarinejad2003}
K.~Shoarinejad, J.~Speyer, and G.~Pottie, ``Integrated predictive power control
  and dynamic channel assignment in mobile radio systems,'' {\em IEEE
  Transactions on Wireless Communications}, vol.~2, no.~5, pp.~976--988, 2003.

\bibitem{Vu2023}
T.~X. Vu, S.~Bhandari, M.~Minardi, H.~Van~Nguyen, and S.~Chatzinotas, ``{3GPP}
  new radio precoding in {NGSO} satellites: Channel prediction and dynamic
  resource allocation,'' in {\em 2023 IEEE Statistical Signal Processing
  Workshop (SSP)}, pp.~115--119, 2023.

\bibitem{3GPP2020}
``Study on new radio ({NR}) to support non-terrestrial networks,'' TR 38.811,
  Release 15, 3rd Generation Partnership Project, 2020.

\bibitem{ITU676}
``Attenuation by atmospheric gases and related effects.,'' ITU-R P.676,
  International Telecommunication Union, 2022.

\bibitem{Jakes1994}
W.~C. Jakes and D.~C. Cox, {\em Microwave Mobile Communications}.
\newblock Wiley-IEEE Press, 1994.

\bibitem{Patzold2006}
M.~Patzold and B.~Hogstad, ``Two new methods for the generation of multiple
  uncorrelated {R}ayleigh fading waveforms,'' in {\em IEEE VTC}, 2006.

\bibitem{Boyd2023}
S.~P. Boyd and L.~Vandenberghe, {\em Convex optimization}.
\newblock Cambridge University Press, 2004.

\bibitem{Chiang2007}
M.~Chiang {\em et~al.}, ``Power control by geometric programming,'' {\em IEEE
  Trans. Wireless Commun.}, vol.~6, pp.~2640--2651, July 2007.

\bibitem{Zhang2021}
Y.~Zhang, Y.~Wu, A.~Liu, X.~Xia, T.~Pan, and X.~Liu, ``Deep learning-based
  channel prediction for {LEO} satellite massive {MIMO} communication system,''
  {\em IEEE Wireless Commun. Lett.}, vol.~10, Aug. 2021.

\bibitem{Liao2015}
Q.~Liao, S.~Valentin, and S.~Stanczak, ``Channel gain prediction in wireless
  networks based on spatial-temporal correlation,'' in {\em IEEE SPAWC}, 2015.

\bibitem{Nocedal2006}
S.~J.~W. Jorge~Nocedal, {\em Numerical Optimization}.
\newblock Springer New York, 2006.
\newblock Chapter 14.

\bibitem{Dahl2021}
J.~Dahl and E.~D. Andersen, ``A primal-dual interior-point algorithm for
  nonsymmetric exponential-cone optimization,'' {\em Mathematical Programming},
  2021.

\bibitem{Fontan2001}
F.~Fontan, M.~Vazquez-Castro, C.~Cabado, J.~Garcia, and E.~Kubista,
  ``Statistical modeling of the {LMS} channel,'' {\em IEEE Transactions on
  Vehicular Technology}, vol.~50, no.~6, pp.~1549--1567, 2001.

\bibitem{Zhang2022}
Y.~Zhang, A.~Liu, P.~Li, and S.~Jiang, ``Deep learning ({DL})-based channel
  prediction and hybrid beamforming for {LEO} satellite massive {MIMO}
  system,'' {\em IEEE Internet of Things Journal}, vol.~9, no.~23,
  pp.~23705--23715, 2022.

\bibitem{Ying2024}
M.~Ying, X.~Chen, Q.~Qi, and W.~Gerstacker, ``Deep learning-based joint channel
  prediction and multibeam precoding for {LEO} satellite internet of things,''
  {\em IEEE Transactions on Wireless Communications}, vol.~23, no.~10,
  pp.~13946--13960, 2024.

\bibitem{Zhao2025}
C.~Zhao, Y.~He, and X.~Wang, ``Parallel forecasting of {LEO} satellite downlink
  channels using a {CNN}-transformer architecture,'' in {\em 2025 IEEE 7th
  International Conference on Communications, Information System and Computer
  Engineering (CISCE)}, pp.~97--101, 2025.

\bibitem{Jain1984}
R.~Jain, D.-M. Chiu, and W.~Hawe, ``A quantitative measure of fairness and
  discrimination for resource allocation in shared computer system,'' {\em
  Eastern Research Laboratory, Digital Equipment Corporation}, 1984.

\bibitem{Proakis2007}
J.~Proakis and P.~Massoud~Salehi, {\em Digital Communications}.
\newblock McGraw-Hill Education, 2007.

\end{thebibliography}
\bibliographystyle{ieeetr}

\end{document}